\documentclass[reqno,12pt,centertags]{amsart}
\usepackage{amsmath,amsthm,amscd,amssymb,latexsym,upref,stmaryrd,epsfig,graphicx}
\usepackage[cp1251]{inputenc}
\usepackage[english]{babel}
\usepackage{mathtext}
\usepackage{amsfonts}
\usepackage{graphicx}
\usepackage{epsfig}

\usepackage{hyperref}
\usepackage[numbers,sort&compress]{natbib}
\newcommand*{\mailto}[1]{\href{mailto:#1}{\nolinkurl{#1}}}
\usepackage{amsmath}
\usepackage{amsthm}
\newtheorem{corollary}{Corollary}[section]
\newtheorem{theorem}{Theorem}[section]

\newtheorem{remark}{Remark}[section]

\numberwithin{equation}{section}

\begin{document}

\thispagestyle{empty}

\noindent{\large\bf Determination of matrix exponentially decreasing potential from scattering matrix}
\\

\noindent {\bf  Xiao-Chuan Xu}\footnote{Department of Applied
Mathematics, School of Science, Nanjing University of Science and Technology, Nanjing, 210094, Jiangsu,
People's Republic of China, {\it Email:
xcxu@njust.edu.cn}}
{\bf and Chuan-Fu Yang}\footnote{Department of Applied
Mathematics, School of Science, Nanjing University of Science and Technology, Nanjing, 210094, Jiangsu,
People's Republic of China, {\it Email: chuanfuyang@njust.edu.cn}}
\\

\noindent{\bf Abstract.}
{(i) For the matrix Schr\"{o}dinger operator on the half line, it is shown that if the potential exponentially decreases fast enough then only the scattering matrix uniquely determines the self-adjoint potential and the boundary condition. (ii) For the matrix Schr\"{o}dinger operator on the full line, it is shown that if the potential exponentially decreases fast enough then  the scattering matrix (or equivalently, the transmission coefficient and reflection coefficient) uniquely determine the potential. If the potential vanishes on $(-\infty,0)$ then only the left reflection coefficient uniquely determine the potential.}



\section{Introduction and main result}
Consider the matrix Schr\"{o}dinger equation
\begin{equation}\label{1}
-\Psi''(x)+V(x)\Psi(x)=k^2\Psi(x),\quad x\in I,
\end{equation}
where $I$ denotes $(-\infty,\infty)$ or $(0,\infty)$, $\Psi$ is either an $n\times n$ matrix-valued function or a column vector-valued function with $n$ components, and
the matrix potential $V(x)$ satisfies
\begin{equation}\label{2}
V(x)^\dag=V(x)\quad \text{and}\quad
\sum_{1\le l ,s\le n}\int_I(1+|x|)| V_{ls}(x)|dx<\infty
\end{equation}
Here the dagger "$\dag$" denotes the matrix adjoint (complex conjugate and matrix transpose).

The matrix Schr\"{o}dinger equation (\ref{1}) has been studied by many authors \cite{AM,AW,AKW,MH1,MH3,XY}, which is connected with scattering in quantum mechanics involving particles of internal structures as spins, scattering on graphs and quantum wires (see \cite{KS1,KS2,PK} and the references therein), and has important applications in matrix KdV and Boomeron equations \cite{CD}.

For the self-adjoint matrix Schr\"{o}dinger operator on the half line, in order to recover the potential and the parameter in the boundary condition, in general, it is necessary to prescribe the scattering data, which consists of scattering matrix and bound state data (eigenvalues and normalization matrices) \cite{MH3,XY}.
 In the scalar case ($n=1$ in (\ref{1})) with Dirichlet boundary condition, there is an interesting result: if the potential is compactly supported, then only the scattering matrix can uniquely determine the potential \cite{RM}. In this paper, we generalize this result to the matrix case with the general self-adjoint boundary condition.

  For the full line case, it is also known  \cite{NB1} that in order to recover the potential one has to specify the right (or left) reflection coefficient and bound state data (eigenvalues and wight matrices). We shall prove that if the potential exponentially decreases fast enough then the transmission coefficient and right (or left) reflection coefficient uniquely determine the potential. It was shown in \cite{NB,FY} that the Weyl matrix uniquely determine the potential on the half line. We will use this result to show that if the potential vanishes on $(-\infty,0)$ then only the left reflection coefficient uniquely determine the potential.

  These results are of interest because they indicate that in some cases the self-adjoint matrix potential can be recovered uniquely without knowing the normalization matrices or weight matrices, which, in general, have no physical meaning. Moreover, the scattering matrix are much more amenable to direct measurement than
bound state data in scattering
experiments, and it is effective in numerical reconstructions \cite{RS}.

\section{Inverse scattering on the half line}
In this section we
consider matrix Schr\"{o}dinger equation (\ref{1}) on the half line (i.e., $I=\mathbb{R}^+:=(0,\infty)$) with the general self-adjoint boundary condition \cite{AKW,MH1,MH3}
\begin{equation}\label{3}
 -B^\dag \Psi(0)+A^\dag \Psi'(0)=0_n,
\end{equation}
where $0_n$ denotes  zero matrix or  zero vector, and
\begin{equation}\label{4}
A=\frac{1}{2}\left(U+I_n\right)\quad\text{and}\quad B=\frac{i}{2}\left(U-I_n\right),
\end{equation}
the matrix $U$ is unitary, and $I_n$ denotes the  $n \times n$ identity matrix.
Denote by $L(V,U)$ the problem (\ref{1}) and (\ref{3}).

Let us recall that the \emph{eigenvalue} of the problem $L(V,U)$ is the $k^2$-value for which (\ref{1}) has a nonzero column vector solution $\Psi\in L^2(\mathbb{R}^+)$ satisfying (\ref{3}). It is shown in \cite{AM,AW,MH3} that the eigenvalues of the problem $L(V,U)$ correspond to the zeros of determinant of the Jost matrix $J(k)$ on the positive half imaginary axis $i\mathbb{R}^+$.  Here the Jost matrix $J(k)$ is defined as \cite{AM,AKW,AW,MH1,MH3}
\begin{equation}\label{12}
J(k):=f_+(-\bar{k},0)^\dagger B-f_+'(-\bar{k},0)^\dagger A,\quad k\in \overline{\mathbb{C}}^+:=\{k:{\rm Im}k\ge0\},
\end{equation}
where $\bar{k}$ means the conjugate of the number $k$, and the matrix-valued function $f(k,x)$ is the Jost solution to (\ref{1}) satisfying the
integral equation
\begin{equation}\label{jost}
  f_+(k,x)=e^{ikx}I_n+\int_x^\infty\frac{\sin k(t-x)}{k}V(t)f_+(k,t)dt,\quad k\in \overline{\mathbb{C}}^+.
\end{equation}
Denote by $\{ik_j\}_{j=1}^N$ the zeros of $\det J(k)$ on $i\mathbb{R}^+$ with $k_j<k_{j+1}$.
Define the scattering matrix \cite{AKW,AW,MH1,MH3,XY}
\begin{equation}\label{15}
S(k)=    -J(-k)J(k)^{-1},\;\;k\in\mathbb{R},
\end{equation}
and the normalization matrices  \cite{MH3,XY}
\begin{equation}\label{34}
C_j:=P_j \left(P_j\int_0^\infty f_+(ik_j,x)^\dagger f_+(ik_j,x)dx P_j+I_n-P_j\right)^{-\frac{1}{2}},
\end{equation}
where $P_j$ is the orthogonal projection onto $\ker J(ik_j)^\dagger$.  The data
$$\mathcal{S}:=\{S(k),k_j,C_j\}_{k\in \mathbb{R};j=\overline{1,N}}$$ is called \emph{scattering data}.

The main result in this section is as follows.

\begin{theorem}
Assume that the self-adjoint potential $V(x)$ satisfies
 \begin{equation}\label{hl}
   \sum_{1\le l ,s\le n}\int_0^\infty| V_{ls}(x)|e^{2\gamma x}dx<\infty
 \end{equation}
 with $\gamma>k_N$. Then the scattering matrix $S(k)$ uniquely determines the potential $V(x)$ and the unitary matrix $U$ in the boundary condition.
\end{theorem}

\begin{remark}
Note that the Jost matrix $J(k)$ uniquely determines the scattering matrix $S(k)$, thus the Jost matrix $J(k)$ uniquely recovers the problem $L(V,U)$ with the potential $V$ satisfying (\ref{2}). Apparently, the same result holds for the compactly supported potential.
\end{remark}

\begin{remark}
Theorem 1.1 shows that the scattering matrix $S(k)$ can determines the type of the boundary condition (i.e., Dirichlet or non-Dirichlet), whereas, it is shown in \cite{ASU} that this is not true in the scalar case. This is because that the definition of the scattering matrix in \cite{ASU} is different from (\ref{15}) here.
\end{remark}

\begin{proof}[Proof of Theorem 1.1]
It is shown in \cite{XY} that the scattering data $\mathcal{S}$ uniquely determines the potential $V(x)$ satisfying (\ref{2}) and the unitary matrix $U$ in (\ref{4}).
Obviously, the condition (\ref{hl}) implies (\ref{2}).
Thus it is sufficient to prove that the scattering matrix $S(k)$ can uniquely recover the eigenvalues $ \{-k_j^2\}_{j=1}^N$ and the normalization matrices $ \{C_j\}_{j=1}^N$.

Under the condition (\ref{2}),
using the method of successive approximations for (\ref{jost}) (see, e.g., Theorem 2.1.1 in \cite{FY1}), one can easily get that the Jost solution $f_+(k,x)$ and its derivative $f_+'(k,x)$ have analytic continuations for $k$ from $\overline{\mathbb{C}}^+$ to $\Omega_\gamma:=\{k\in\mathbb{C}:{\rm Im}k>-\gamma\}$, and satisfy the asymptotics
\begin{equation}\label{yyy}
 f_+^{(v)}(k,x)=(ik)^v e^{ikx}[1+o(1)],\; x\to+\infty,\; v=0,1,\;k\in\Omega_\gamma\setminus\{0\}.
\end{equation}
Note that the analyticities of the matrix-valued functions $f_+^{(v)}(-\bar{k},0)^\dag$ ($v=0,1$) coincide with the  analyticities of $f_+^{(v)}({k},0)$, respectively. Thus the Jost matrix $J(k)$ has an analytic continuation from $\overline{\mathbb{C}}^+$ to $\Omega_\gamma$, which implies from (\ref{15}) that the scattering matrix $S(k)$ has an analytic continuation from $\mathbb{R}$ to $\{k\in \mathbb{C}: |{\rm Im}k|<\gamma\}$. On the other hand, it is known \cite{AW,MH3} that
the inverse of the Jost matrix $J(k)^{-1}$ has simple poles at $k=ik_j$ ($j=\overline{1,N}$), namely,
\begin{equation}\label{13}
 J(k)^{-1}=\frac{N_{-,k_j}}{k-ik_j}+N_{0,k_j}+O(k-ik_j),\quad k\to ik_j\; \text{ in }\; \overline{\mathbb{C}}^+, \;j=\overline{1,N},
\end{equation}
where $N_{-,k_j}(\ne0_n)$ and $N_{0,k_j}$ are constant matrices.

\textbf{Step 1.} $S(k)$ determines $\{k_j\}_{j=1}^N$.

From (\ref{15}) and (\ref{13}), and the assumption $\gamma>k_N$, we get that the scattering matrix $S(k)$ has simple poles $\{ik_j\}_{j=1}^N$ if $J(-ik_j)\ne0_n$ for all $j=\overline{1,N}$.
Moreover, all the eigenvalues $\{-k_j^2\}_{j=1}^N$ can be found by observing  $S(k)$ when $k$ moves along the positive imaginary axis (if $k\to ik_j$ then there exists at least one element in $S(k)$ approaching $\infty$).

\textbf{Step 2.} $S(k)$ determines $\{C_j\}_{j=1}^N$.

Let $\varphi(k,x)$ be the solution to (\ref{1}) satisfying the initial condition
\begin{equation}\label{j3}
  \varphi(k,0)=A,\quad \varphi'(k,0)=B.
\end{equation}
From Eq.(3.19) in \cite{XY} (or Eq.(21) in \cite{MH3} with a minor revision), we see that
\begin{equation}\label{bnu}
 2ik_j\varphi(ik_j,x)N_{-,k_j}=-f_+(ik_j,x)C_j^2,\quad j=\overline{1,N}.
\end{equation}
By virtue of (\ref{12}), (\ref{15}) and (\ref{13}) we have
\begin{equation}\label{sox}
\lim_{k\to ik_j}(k-ik_j)S(k)=-J(-ik_j)N_{-,k_j}=f_+'(-ik_j,0)^\dag AN_{-,k_j}-f_+(-ik_j,0)^\dag BN_{-,k_j}.
\end{equation}
From Eqs.(\ref{j3}) and (\ref{bnu}), we get
\begin{equation}\label{mlzx}
AN_{-,k_j}=-\frac{1}{2ik_j}f_+(ik_j,0)C_j^2,\;\;\;BN_{-,k_j}=-\frac{1}{2ik_j}f_+'(ik_j,0)C_j^2.
\end{equation}
Substituting (\ref{mlzx}) into (\ref{sox}), we obtain
\begin{equation}\label{sox1}
\lim_{k\to ik_j}(k-ik_j)S(k)=\frac{1}{2ik_j}[f_+(-ik_j,0)^\dag f_+'(ik_j,0)-f_+'(-ik_j,0)^\dag f_+(ik_j,0)]C_j^2.
\end{equation}

Due to that the Jost solution $f_+(k,x)$ with $k\in\Omega_\gamma$ satisfies (\ref{1}) and the potential $V(x)=V(x)^\dag$, then $f_+(\overline{k},x)^\dag$ satisfies
\begin{equation}\label{5}
-f_+''(\bar{k},x)^\dag+f_+(\bar{k},x)^\dag V(x)=k^2f_+(\bar{k},x)^\dag,\quad x\ge0,\;\;{\rm Im}k\le \gamma.
\end{equation}
Using (\ref{1}) and (\ref{5}) it is easy to prove that the Wronskian
\begin{equation}\label{xcx0}
  [f_+(\bar{k},x)^\dag;f_+(k,x)]:=f_+(\bar{k},x)^\dag f_+'(k,x)-f_+'(\bar{k},x)^\dag f_+(k,x)
\end{equation}
is independent of $x$. This yields
\begin{equation*}
  [f_+(\bar{k},x)^\dag;f_+(k,x)]_{x=0}=[f_+(\bar{k},x)^\dag;f_+(k,x)]_{x=\infty},\;\;|{\rm Im}k|\le \gamma.
\end{equation*}
Letting $x\to\infty$ in (\ref{xcx0}) and using the asymptotics (\ref{yyy}) with $k=\pm ik_j$, we get
\begin{equation}\label{bgha}
[f_+(-ik_j,x)^\dag;f_+(ik_j,x)]_{x=0}=-2k_j,\quad j=\overline{1,N}.
\end{equation}
Together with (\ref{sox1}) and (\ref{bgha}), we obtain
\begin{equation}\label{jja}
\lim_{k\to ik_j}(k-ik_j)S(k)=iC_j^2,\quad j=\overline{1,N},
\end{equation}
that is,
\begin{equation}\label{jja1}
C_j^2=-i\mathop {\rm Res }\limits_{k= ik_j}S(k),\quad j=\overline{1,N}.
\end{equation}

\textbf{Step 3.}
One proves $J(-ik_j)\ne0_n$ for $j=\overline{1,N}$.

It is known \cite{MH3,XY} that the normalization matrices $C_j$ are all non-zero. It follows from (\ref{15}) and (\ref{jja}) that $J(-ik_j)\ne0_n$.

From the steps 1-3, the proof is finished.
\end{proof}

\begin{remark}
In the proof of Theorem 1.1 we actually give the reconstruction algorithm for recovering the problem $L(V,U)$ from the scattering matrix $S(k)$:

 1) find the eigenvalues by observing $ S(k)$ on $i\mathbb{R}^+$;

 2) use the eigenvalues obtained above to recover the normalization matrices by Eq.(\ref{jja1});

 3) use the Marchenko procedure (see \cite{MH3,XY}) to recover the potential $V$ and the unitary matrix $U$ in the boundary condition.
\end{remark}

\section{Inverse scattering on the full line}
In this section, let us consider the Schr\"{o}dinger equation (\ref{1}) on the real line (i.e., $I=\mathbb{R}$). Together with the Jost solution $f_+(k,x)$, we consider another Jost solution $f_-(k,x)$ which satisfies
\begin{equation}\label{jost1}
   f_-(k,x)=e^{-ikx}I_n+\int_{-\infty}^x\frac{\sin k(x-t)}{k}V(t)f_-(k,t)dt,\quad k\in \overline{\mathbb{C}}^+.
\end{equation}

Let us study the Schr\"{o}dinger equation (\ref{1}) under the assumption
 \begin{equation}\label{fll}
   \sum_{1\le l ,s\le n}\int_0^\infty| V_{ls}(x)|e^{2\gamma_1 |x|}dx<\infty,
 \end{equation}
where $\gamma_1>0$.
Similarly, under the condition (\ref{fll}), one can get that  $f_-^{(v)}(k,x)$ ($v=0,1$) have analytic continuations for $k$ from $\overline{\mathbb{C}}^+$ to $\Omega_{\gamma_1}:=\{k\in\mathbb{C}:{\rm Im}k>-\gamma_1\}$, and satisfy the asymptotics
\begin{equation}\label{fl0}
 f_-^{(v)}(k,x)=(-ik)^v e^{-ikx}[1+o(1)],\; x\to-\infty,\; v=0,1,\;k\in\Omega_{\gamma_1}\setminus\{0\}.
\end{equation}

Let us recall some physical quantities \cite{AKV,NB1}.
Denote
\begin{align}
 A(k):&=\frac{1}{2ik}[f_-(-\bar{k},x)^\dag;f_+(k,x)],\quad k\in\Omega_{\gamma_1}\setminus\{0\}, \label{fl1}\\
 B(k):&=-\frac{1}{2ik}[f_-(\bar{k},x)^\dag;f_+(k,x)],\quad |{\rm Im}k|\le \gamma_1,k\ne0,\label{fl2}\\
 C(k):&=-B(\bar{k})^\dag, \;\;D(k):=A(-\bar{k})^\dag,
\end{align}
and
\begin{equation}\label{fl3}
\left\{\begin{split}
S_-(k)&=B(k)A(k)^{-1}, \; S_+(k)=C(k)D(k)^{-1} ,\; |{\rm Im} k|\le \gamma,\\
T_+(k)&=A(k)^{-1},\;T_-(k)=D(k)^{-1},\; k\in\Omega_{\gamma_1}\setminus\{0\}.
\end{split}\right.
\end{equation}
The matrices $S_\pm(k)$ and $T_\pm(k)$ are called the \emph{reflection coefficients} and the\emph{ transmission coefficients}, respectively.
The $2n\times 2n$ matrix
\begin{equation}\label{fl4}
  \begin{bmatrix}
  T_+(k)& S_+(k)\\
  S_-(k)& T_-(k)
  \end{bmatrix}
\end{equation}
is called the \emph{scattering matrix}. The \emph{eigenvalue} of the Schr\"{o}dinger operator (\ref{1}) is the $k^2$-value for which (\ref{1}) has a nonzero column vector solution $\Psi\in L^2(\mathbb{R})$, which coincides with the zero of the determinant $\det A(k)$ on the positive imaginary axis $i\mathbb{R}^+$ \cite{NB1}. Denote, also, by $\{ik_j\}_{j=1}^N$ the zeros of $\det A(k)$ on $i\mathbb{R}^+$ with $k_j<k_{j+1}$. It is known \cite{NB1} that the matrices $T_\pm(k)$ have simple poles at $\{ik_j\}_{j=1}^N$, moreover, if we denote
\begin{equation}\label{fl5}
  R_j^\pm:=\mathop {\rm Res }\limits_{k= ik_j}T_\pm(k),
\end{equation}
then there exist positive semidefinite matrices $N_j^\pm$ such that
\begin{equation}\label{fl6}
f_\pm(ik_j,x)R_j^\mp=if_\mp(ik_j,x)N_j^\mp.
\end{equation}
Here the matrices $N_j^-$ and $N_j^+$ are called the left and right \emph{weight matrices}, respectively. The collections
\begin{equation*}
 \mathcal{ S}_-:=\{S_-(k), k_j, N_j^-\}_{k\in \mathbb{R},j=\overline{1,N}},\quad \mathcal{ S}_+:=\{S_+(k), k_j, N_j^+\}_{k\in \mathbb{R},j=\overline{1,N}}
\end{equation*}
are called the left and the right \emph{scattering data}, respectively.

The first main result in this section is as follows.

\begin{theorem}
Assume that the self-adjoint potential $V(x)$ satisfies (\ref{fll}) with $\gamma> k_N$. Then the scattering matrix (\ref{fl4}) uniquely determine the potential.
\end{theorem}

\begin{remark}
Similar to the half line case, the same result is true for the compactly supported potential.
\end{remark}

\begin{proof}[Proof of Theorem 3.1]

It was shown in \cite{NB1} that the potential $V(x)$ satisfying (\ref{2})
is uniquely determined by
the left (or right) scattering data. We shall prove that the scattering matrix uniquely determine the bound state data (i.e., eigenvalues and weight matrices).

Firstly, the eigenvalues $\{-k_j^2\}_{j=1}^N$ can be found by observing  $T_\pm(k)$ when $k$ moves along the positive imaginary axis (if $k\to ik_j$ then there exists at least one element in $T_\pm(k)$ approaching $\infty$).

Next, let us show that the scattering matrix uniquely determine the left weight matrices. By virtue of (\ref{fl3}), we can obtain the matrices $A(k)$ and $B(k)$ from the scattering matrix. From (\ref{fl6}), we have
\begin{equation}\label{fl7}
\left\{ \begin{split}
  f_-(-ik_j,x)^\dag f_+'(ik_j,x)R_j^-=if_-(-ik_j,x)^\dag f_-'(ik_j,x)N_j^-,\\
    f_-'(-ik_j,x)^\dag f_+(ik_j,x)R_j^-=if_-'(-ik_j,x)^\dag f_-(ik_j,x)N_j^-,
 \end{split}\right.
\end{equation}
which implies
\begin{equation}\label{fl8}
 [ f_-(-ik_j,x)^\dag;f_+(ik_j,x)]R_j^-=i[f_-(-ik_j,x)^\dag;f_-(ik_j,x)]N_j^-.
\end{equation}
Using (\ref{fl0}) with $k=\pm ik_j$, we get
\begin{equation}\label{fl9}
[f_-(-ik_j,x)^\dag;f_-(ik_j,x)]=2k_jI_n.
\end{equation}
Together with (\ref{fl2}), (\ref{fl8}) and (\ref{fl9}), we have
\begin{equation}\label{fh10}
 -iB(ik_j)R_j^-=N_j^-.
\end{equation}
Since the matrix $R_j^-$ can be obtained from $T_(k)$ (see (\ref{fl5})), we conclude that the scattering matrix uniquely recover the left scattering data $\mathcal{S}_-$. The proof is complete.
\end{proof}
\begin{remark}
The proof of Theorem 3.1 provides a reconstruction algorithm for recovering the self-adjoint potential satisfying (\ref{fll}) from the scattering matrix (\ref{fl4}):

1) find the eigenvalues from $T_-(k)$ (or $T_+(k)$);

2) find the left weight matrices $N_j^-$ by (\ref{fh10});

3) use the left scattering data $\mathcal{S}_-$ to recover the potential by the method in \cite{NB1}.
\end{remark}

Now, we state the second main result in this section.

\begin{theorem}
Assume that the potential $V(x)$ satisfies (\ref{2}) and vanishes on $(-\infty,0)$, then the  reflection coefficient $S_-(k)$ uniquely determine the potential.
\end{theorem}
\begin{proof}
Denote
\begin{equation}\label{f11}
  f_{0\pm}(k,x)=  f_\pm(k,x)T_\pm(k),\quad x\in \mathbb{R}.
\end{equation}
It is known \cite{NB1} that
\begin{equation}\label{fl11}
  f_{0\pm}(k,x)=f_\mp(- k,x)+f_\mp(k,x)S_\mp(k),\quad x\in \mathbb{R}.
\end{equation}
Since $V(x)=0$ for $x<0$, we have 
\begin{equation}\label{fl12}
  f_{0+}(k,x)=e^{ikx}I_n+e^{-ikx}S_-(k),\quad x\le0.
\end{equation}
It follows from (\ref{f11}) and (\ref{fl12}) that
\begin{equation*}
  f_+(k,0)f_+'(k,0)^{-1}=[I_n+S_-(k)][I_n-ikS_-(k)]^{-1}.
\end{equation*}
 Note that the matrix $f_+(k,0)f_+'(k,0)^{-1}$ is the Weyl matrix, which uniquely determines the potential $V(x)$ on $(0,\infty)$ (see \cite{NB}).
Thus, the reflection coefficient $S_-(k)$ uniquely determine the potential $V(x)$ vanishing on $(-\infty,0)$. The proof is complete.
\end{proof}

\begin{corollary}
Assume that the potential $V(x)$ satisfies (\ref{2}) and vanishes on $(-\infty,0)$, then the matrix $f_{0+}(k,0)$ (or $f_{0+}'(k,0)$) uniquely determine the potential.
\end{corollary}

\noindent {\bf Acknowledgments.}
The research work was supported in part by the National Natural Science Foundation of China (11171152 and 91538108) and
Natural Science Foundation of Jiangsu Province of China (BK 20141392).


\begin{thebibliography}{99}\footnotesize
\bibitem{AM} Z.S. Agranovich, V.A. Marchenko, The Inverse Problem of Scattering Theory, Gordon and Breach, New York,
1963.
\bibitem{AKV} T. Aktosun, M. Klaus,  C. van der Mee, Small-energy asymptotics of the scattering matrix for the matrix Schr\"{o}dinger equation
on the line, J. Math. Phys. 42 (2001), 4627.
\bibitem{AKW} T. Aktosun, M. Klaus, R. Weder, Small-energy analysis for the self-adjoint matrix Schr\"{o}dinger operator on the half
line, J. Math. Phys. 52 (2011), 102101.
\bibitem{AW} T. Aktosun, R. Weder, High-energy analysis and Levinson's theorem for the self-adjoint matrix Schr\"{o}dinger operator
on the half line, J. Math. Phys. 54 (2013), 012108.
\bibitem{ASU} T. Aktosun, P. Sacks,  M. Unlu, Inverse problems for selfadjoint Schr\"{o}dinger operators
on the half line with compactly supported potentials, J. Math. Phys. 56 (2015), 022106.
\bibitem{NB1} N. Bondarenko, Inverse scattering on the line for the matrix Sturm-Liouville equation, J. Differential Equations 262 (2017), 2073-2105.
\bibitem{NB} N. Bondarenko, An inverse spectral problem for the matrix
Sturm-Liouville operator on the half-line, Boundary Value Problems 15 (2015), 22pp.
\bibitem{CD} F. Calogero, A. Degasperis, Nonlinear evolution equations solvable by the inverse spectral transform II, Nuovo
Cimento B 39 (1) (1977).
\bibitem{FY} G. Freiling, V. Yurko, An inverse problem for the non-selfadjoint
matrix Sturm-Liouville equation on the half-line, J. Inv. Ill-Posed Problems 15 (2007), 785-798.
  \bibitem{FY1} G. Freiling, V.A. Yurko, Inverse Sturm-Liouville Problems and Their Applications,
NOVA Science Publishers, New York, 2001.
\bibitem{MH1} M.S. Harmer, Inverse scattering for the matrix Schr\"{o}dinger operator and Schr\"{o}dinger operator on graphs with general
self-adjoint boundary conditions, ANZIAM J. 44 (2002), 161-168.
\bibitem{MH3} M.S. Harmer, Inverse scattering on matrices with boundary conditions, J. Phys. A: Math. Gen. 38 (2005) 4875-4885.
\bibitem{KS1} V. Kostrykin, R. Schrader, Kirchhoff's rule for quantum wires. II: The inverse problem with possible applications to
quantum computers, Fortschr. Phys. 48 (2000), 703-716.
\bibitem{KS2} P. Kurasov, F. Stenberg, On the inverse scattering problem on branching graphs, J. Phys. A 35 (2002), 101-121.
\bibitem{PK} P. Kuchment, Quantum graphs. I. Some basic structures, Waves Random Media 14 (2004), S107-S128.
\bibitem{RM} A.G. Ramm, One-dimensional inverse scattering and spectral problems, CUBO a Mathematical Journal 6  (2004), 313-426.
\bibitem{RS} W. Rundell, P. Sacksb, On the determination of potentials without bound state data, J. Comput. Appl. Math. 55 (1994), 325-347.
\bibitem{XY} X.-C. Xu, C.-F. Yang, Solvability of the inverse scattering problem for the self-adjoint matrix Schr\"{o}dinger operator on the half line  (Preprint, arXiv:1703.01375v1 [math-ph]).
\end{thebibliography}
\end{document}